\documentclass[final,5p, twocolumn]{elsarticle}
\usepackage{graphicx}
\usepackage{amsthm}
\theoremstyle{definition}
\newtheorem{definition}{Definition}[section]
 \newtheorem{proposition}{Proposition}[section]
\usepackage{diagbox}
 \usepackage{algorithm2e}
\usepackage{url}

\usepackage{appendix}
\usepackage{hyperref}
\usepackage{graphicx}
\usepackage{amssymb}
 \usepackage[]{algorithm2e}
 \usepackage{mathtools}
\journal{Journal Name}

\begin{document}
\setlength{\parindent}{0pt}
\begin{frontmatter}

\title{The $k$-interchange-constrained diameter of a transit network: A connectedness indicator that accounts for travel convenience}


\author{Nassim Dehouche}
\address{Business Administration Division,\\
Mahidol University International College\\
Salaya, 73170, Thailand\\
Nassim.deh@mahidol.ac.th }
\begin{abstract}
We study two variants of the shortest path problem. Given an integer $k$, the \textit{$k$-color-constrained} and the \textit{$k$-interchange-constrained} shortest path problems, respectively seek a shortest path that uses no more than $k$ colors and one that makes no more than $k-1$ alternations of colors. We show that the former problem is NP-hard, when the latter is tractable. The study of these problems is motivated by some limitations in the use of diameter-based metrics to evaluate the topological structure of transit networks. We notably show that indicators such as the diameter or directness of a transit network fail to adequately account for travel convenience in measuring the connectivity of a network and propose a new network indicator, based on solving the \textit{$k$-interchange-constrained} shortest path problem, that aims at alleviating these limitations. \\


\end{abstract}
\end{frontmatter}

\section*{Acknowledgement}
The author would like to thank Dr. Yuval Filmus for his helpful advice, as well as two anonymous reviewers for their valuable feedback that greatly improved this paper. 
\section{Introduction and Related Work}
Cities and metropolitan areas, are characterized by a constant growth and evolution of their demographics and economies. Thus, transit networks should be seen as constantly expanding objects and transportation planning as a highly strategic, and dynamic decision process that requires taking multiple factors into account \cite{Chile}. A user-based perspective is often favored by planners \cite{Guihaire}, through the optimization of aspects such as passenger flows \cite{Yu}, demographic coverage \cite{Kansky}, travel time \cite{Zhao}, or waiting time \cite{Wong}. This user-based perspective is typically adjusted with cost \cite{Chile2}, or robustness considerations \cite{Cadarso} as well as through the assessment of the impact of potential network designs on private vehicle traffic \cite{Traf}, the environment \cite{Env} and their integration with other modes of public transportation \cite{Peng}. These traditional planning factors can be supplemented by more descriptive indicators reflecting the topological properties of transit networks. This emerging, relatively under-utilized approach to transport planning \cite{der}, whose development has closely followed advances in network science \cite{Network} and geographic information science \cite{GIS} can aid transportation planners in reaching a better understanding of existing network layouts and anticipating their expansion. Three aspects of a transportation network can be evaluated by these indicators, according to the typology of \cite{der}; \textit{state}, \textit{form}, and \textit{structure}. Our contributions in this paper focus on the evaluation of the latter aspect, \textit{structure}, that is the intrinsic topological properties of existing networks. Among other classical network indicators \cite{Kansky} that are commonly used in this type of analysis, the diameter $d$ of a network is defined as the length of its longest shortest path. The smaller its diameter, the better linked \cite{geo} the network. Better networks, from a passenger's perspective, would be expected to be long (i.e. the sum of lengths of their edges should be large), while having a small diameter to reduce travel time. The network extension $\Pi$, defined as the ratio between a network's total length and its diameter, reflects this preference and is a measure of a network's spread \cite{app}. The higher the value of $\Pi$, the more widespread a network tends to be.\\ 

In this work, we argue that a limitation of both of these indicators is that they do not account for a network's division in transit lines, which can be a defining aspect for travel convenience. Indeed, for a network of a certain diameter and total length, a passenger would obviously prefer to be able to cover these two distances while making a smaller number of line interchanges. Thus there can exist important disparities in terms of travel convenience in networks presenting equal values for these two classical indicators. \\

Derrible and Kennedy \cite{der} have recently made a similar point and introduced a new network indicator known as directness, noted $\tau$, which is a measure of the "ease of travel" within a network to avoid unnecessary interchanges, and is defined as the ratio of the number of lines to the maximum number of transfers to travel between the two most distant nodes of the network. We show however that this indicator is an insufficient measure of the convenience of transit over a network, with counter-examples presented in Section \ref{4}.\\
 
Section \ref{a} introduces two new network indicators, namely the color-constrained and the interchange-constrained diameters of a graph (respectively denoted $dc$ and $di$), which both correct some inaccuracies resulting from the use of the previously mentioned classical network indicators to evaluate transit networks. We show that for a generic graph, the calculation of $dc$ corresponds to solving an NP-hard extension of the shortest path problem, when $di$ can be computed efficiently by solving a tractable such extension. This section additionally shows that the latter indicator is opportunely more adapted than the former to transit networks. Section \ref{d} concludes this paper with preliminary results stemming from an ongoing application of
the proposed network indicator to real transit networks.

\section{Problem Statement}
\label{a}
Let $G=(V, E)$ be an undirected connected graph, in which $V$ is a set of $n$ nodes, and $E$ a set of $m$ edges. Each edge $e \in E$ possesses two attributes $l(e) \in L$ and $c(e) \in C$, such that $l(e)$ denotes the length of edge $e$ (e.g. as a unitary connection or an actual length in kilometers), and $c(e)$ its color or label.\\

Graph $G$ can be used to model the topological structure of a multiple-line, unimodal transit network, by associating a node $v \in V$ to each station and an edge $e \in E$ to each direct connection between two stations, the color $c(e)$ of which denotes the transit line it belongs to. It should be noted that the term "line" here, is used in its widest sense and could include circular sequences of transit stations (e.g. the "Circle Line" of the London Underground network or the "Line 10 loop" of the Beijing Subway network). These types of circular structures would simply result in monochromatic cycles \cite{erdos} in $G$ . In the remainder of this work, we shall consider, with no loss of generality, that all edges are of length $l(e)=1$. That is to say that we study transit networks from a topological, graph-theoretic perspective, rather than a geographic one \cite{gattuso}. Each concept introduced in this work can nevertheless be trivially extended to the general case, where we account for the actual distances between stations, rather than measuring distances in terms of the number of stations traveled.

\begin{definition}[$\mathcal{PI}_k$: the $k$-interchange-constrained shortest path problem]
For a given integer $k=1, \dots, n-2$ and two nodes $s, t \in V$, the $k$-interchange-constrained shortest path problem consists in determining the shortest Path between $s$ and $t$, with the additional constraint that the number of alternations of colors in the sequence of edges defining a path cannot exceed $k-1$. We call feasible solutions to this problem  $k$-interchange-constrained paths.
\end{definition}

\begin{definition}[$di_k$: the $k$-interchange-constrained diameter of a graph]
For a given integer $k=1, \dots, n-2$, we define the $k$-color-constrained diameter of graph $G$ as the maximum length (with respect to $L$) among all pairs of nodes in $V$, of a shortest $k$-interchange-constrained path linking a pair of nodes. We set the notation $di=di_{n-2}$.\\
\end{definition}

\begin{definition}[$\mathcal{PC}_k$: the $k$-color-constrained shortest path problem]
For a given integer $k=1, \dots, |C|$ and two nodes $s, t \in V$, the $k$-color-constrained shortest path problem consists in determining the shortest Path between $s$ and $t$, with the additional constraint that the number of different colors appearing in the sequence of edges defining a path cannot exceed $k$. We call feasible solutions to this problem  $k$-color-constrained paths.
\end{definition}


\begin{definition}[$dc_k$: the $k$-color-constrained diameter of a graph]
For a given integer $k=1, \dots, |C|$, we define the $k$-color-constrained diameter of graph $G$ as the maximum length (with respect to $L$) among all pairs of nodes in $V$, of the shortest $k$-color-constrained path linking a pair of nodes. We set the notation $dc=dc_{|C|}$.\\
\end{definition}

If $G=(V, E)$ represents a multiple-line unimodal transit network, the $k$-color-constrained diameter of $G$ is the maximum number of stations or distance, that a passenger can travel while using a shortest path of no more than $k$ transit lines, while the $k$-interchange-constrained diameter of $G$ is the maximum number of stations or distance, that a passenger can travel while using a shortest path making no more than $k-1$ interchanges and thus also using no more than $k$ transit lines. \\

An important relation between problems $\mathcal{PC}_k$ and $\mathcal{PI}_k$ to note at this point is that, for a given $k=1, \dots, |C|$ the domain of feasible solutions of $\mathcal{PC}_k$ is a restriction of that of $\mathcal{PI}_k$. Indeed, making a maximum of $k-1$ interchanges implies using at most $k$ colors in a path, but the converse implication is not always true. For instance, a path over which the sequence of colors would be $red-black-red-black$ can be feasible for some instances of $\mathcal{PC}_2$ but it is not feasible for any instance of $\mathcal{PI}_2$, in other words this path is a $2$-color-constrained path, but it is not a $2$-interchange-constrained path, and any $k$-interchange-constrained path is a $k$-color-constrained path. In the perspective of accounting for the convenience of travel when measuring the diameter of a transit network, it is more appropriate to limit the number of interchanges, rather than the number of lines. Indeed, the former approach is a direct reflection of the concerns of passenger, when the latter may give consideration to highly inconvenient paths using a limited number of lines (as in the $red-black-red-black$ example).\\

Problems $\mathcal{PC}_k$ and $\mathcal{PI}_k$ lie at the confluence of two lines of research on two variants of the shortest path problem, without them being specifically treated in the literature, to the best of our knowledge.  Resource constrained shortest path problems \cite{resource} are known to be NP-hard, but neither the constraint in problem $\mathcal{PC}_k$ nor that in problem $\mathcal{PI}_k$ can be expressed as a constraint on an additive resource. On the other hand, research on formal language constrained-path problems \cite{formal} is specifically interested in determining shortest paths in valued, labeled graphs, with the additional constraint that the sequence of labels formed by the edges of a path belongs to a given formal language. An important result in this context is that the problem is solvable efficiently in polynomial time when the language in question is restricted to be a context-free language and NP-hard, when restricted to fixed simple regular language. However, this result does not cover neither problem $\mathcal{PC}_k$ nor problem $\mathcal{PI}_k$, since a constraint on the number of labels, or the number of alternations of labels appearing in a path does not define a regular language, as it cannot be recognized by a finite automaton.

For $k=1$, the $1$-color-constrained and the $1$-interchange-constrained diameters of a transit network $G$ are equal and they do not trivially correspond to the length of the longest transit line, as there can be a shorter way to join the two extremities of a line than going through the whole line, namely if a shorter portion of another line connects these two extremities.

These two indicators rather correspond to the length of the longest portion of a line that constitutes a shortest path using one color, which in and of itself can be an interesting indicator of how well connected the network is, when calculated for each line.\\

\section{Computational complexity}
\label{b}
\label{preuve}
\begin{proposition}
The $k$-color-constrained shortest path problem $\mathcal{PC}_k$ is NP-hard.
\end{proposition}
\begin{proof}
We reduce an instance of the boolean satisfiability problem \cite{cook} to an instance of the decision version of the $k$-color-constrained shortest path problem. Given $F(x_1, x_2,\dots,x_n) = C_1 \land \cdots \land C_m$, a boolean formula in Conjunctive Normal Form, on the variables $x_1,\ldots,x_n$, we construct a graph $G$ containing a node for each variable $x_i, i=1,\dots,n$ and for each clause $C_j, j=1, \dots, m$ and an additional source node $s$. Formally,  $G=(V,E)$, where $V=\{s\}\cup \{x_1,\dots,x_n\}\cup \{C_1,\dots, C_n\}$, and the edges of $E$ are the following:
\begin{itemize}
\item There exist two edges between $s$ and $x_1$, labeled $x_1$ and $\lnot x_1$.
\item There exist two edges between $x_{i-1}$ and $x_i$, labeled $x_i$ and $\lnot x_i$, $\forall i=2,\dots, n$.
\item There exists one edge linking $x_n$ and $C_1$, for each literal in $C_1$, and labeled as that literal.
\item There exists one edge linking $C_{j-1}$ and $C_j$, for each literal in $C_j$, and labeled as that literal, $\forall j=2, \dots, m$.
\end{itemize}

Thus there exist $2n$ labels, one per literal in $F(x_1, x_2,\dots,x_n)$, and we set $k= n$. 
It can be easily seen that $F(x_1, x_2,\dots,x_n)$ is satisfiable if and only if there exists a path from $s$ to $C_m$ which uses at most $n$ labels. Indeed, if such a path exists, it would be of length $2n$ and use exactly $n$ labels all appearing in the first half of the path, going from $s$ to $x_n$. These labels correspond to the literals (each being a variable $x_i$ or its negation $\lnot x_i, i=1,\dots n$) that would take value $1$ in a solution to the boolean satisfiability problem. The existence of the second half of the path, going from $x_n$ to $C_m$ ensures that this assignment of boolean values to the literals satisfies each clause $C_j, j=1, \dots, m$ and therefore that $F(x_1, x_2,\dots,x_n)$ is satisfiable. Conversely, if $F(x_1, x_2,\dots,x_n)$ is satisfiable, then there exists a set of $n$ literals which, when assigned value $1$ satisfy each clause $C_j, j=1, \dots, m$, a path from $s$ to $C_m$ which uses exactly $n$ labels can therefore be built in the graph by selecting a set of $2n$ edges labeled with these $n$ literals. 
\end{proof}

\begin{proposition}
The $k$-interchange-constrained shortest path problem $\mathcal{PC}_k$ is tractable.
\end{proposition}
From a Dynamic Programming perspective, one can see that the shortest $k$-interchange-constrained path problem possesses an optimal sub-structure. Indeed, let $P_{st}^k=s  x_1 x_2  \dots x_{k-2} x_{k-1} t$ be a shortest $k$-interchange-constrained path\footnote{Note that in transit networks, intermediate nodes $x_i, i=1, \dots, k-1$ would necessarily be interchange nodes.} of length $L(P_{st}^k)$, between two nodes $s$ and $t$ in a graph $G$.\\

The functional equation defining problem $\mathcal{PC}_k$ can be broken down as follows:
\begin{itemize}
\item $L(P_{ss}^k)=0$
\item $L(P_{st}^k)=\min\limits_{i\in V-\{s,t\}} \{L(P_{si}^{k-1}) + L(P_{it}^{1}) ) \}$
\end{itemize}
Where $L(P_{si}^{1})$ is the length of a path between $s$ and one of the interchange nodes $x_i$ that is directly accessible from $s$ with edges of one color (i.e. without any prior interchange). In other words, the shortest $k$-interchange-constrained path between $s$ and $t$ contains the shortest path between its first visited interchange node $x_1$ and $t$ using at most $k-1$ interchanges.\\

The previous dynamic programming reasoning can obviously not be used for the NP-hard $\mathcal{PC}_k$ problem, the intuitive reason being that in this problem, any portion of the shortest $k$-color constraint path between $s$ and $t$ can contain up to $k$ different colors and is not necessarily the shortest sub-path to verify this property. \\

The tractability of the $k$-color-constrained shortest path problem is demonstrated in graph-theoretic terms, in the following proof.

\begin{proof}
We reduce an instance of the decision version of the $k$-color-constrained shortest path problem between $s$ and $t$ to an unconstrained shortest path problem. \\

Given an undirected connected graph $G=(V, E)$, in which each edge $e \in E$ possesses two attributes $l(e) \in L$ and $c(e) \in C$, an integer $k=1, \dots, n-2$ and two nodes $s,t \in V$, we construct a directed connected graph $G'=(V', E')$, containing vertices $s$ and $t$ as well as $X_i, i=1, \dots, k-1$, the set of nodes that can be reached from $s$ by a path making exactly $i-1$ alternations of colors. In transit networks this set can be restricted to interchange nodes that can be reached from $s$ by a path making exactly $k-1$ interchanges and in which arcs in $E'$ are only valued by a length $l'(e)$. Formally, we set:
\begin{itemize}
\item $V'=\bigcup \limits_{i=1}^{k-1} X_i \cup \{s,t\}$, where $X_1=\{[x, c]: x \in V, c\in C, \exists P \subseteq E, \mbox{ a path between } s \mbox{ and } x \mbox{ in } G$ $\mbox{ and } \forall e \in P, c(e)=c\}$. An arc $e' \in E'$, from $s$ to $ [x, c]$ is created for each path of the form $P$ in the previous definition of $X_1$. The length of this arc would correspond to the length of path $P$, that is the number of edges of color $c$ in this path, i.e. $l'(e')=\sum \limits _{e \in P} l(e)$, where $ P \subseteq E, \mbox{ is a path between } s \mbox{ and } x \mbox{ in } G \mbox{ and } \forall e \in P, c(e)=c$.
\item $\forall i=2, \dots, k-1, X_i=\{[x, c]: x \in V, c\in C, \exists [x', c'] \in X_{i-1}, \exists P \subseteq E, \mbox{ a path between } x' \mbox{ and } x \mbox{ in } G \mbox{ and }   c(e)=c\neq c', \forall e \in P\} $. 
\end{itemize}

The construction and size of $G'$ are polynomial. Indeed, each set $X_i, i=1, \dots, k-1$ contains a maximum of $|C|\times (n-2)$ nodes and the maximum value of $k$ is $n-2$. Thus there exist at most $|C|\times (n-2)^2 +2$ nodes in graph $G'$, and the number of edges of this graph is asymptotically bounded by $O(n^3)$, since there are $k$ levels ($k$ being bounded by $n-2$) and a maximum of $(n-2)(n-3)$ links between two levels. \\ 

Let $P_{st}^k$ be a $k$-interchange-constrained shortest path between $s$ and $t$, in $G$. We show that $P_{st}^k$ corresponds to a shortest path $P'_{st}$ from $s$ to $t$ in $G'$. In terms of the sequence of colors appearing in $P_{st}^k$, this path can be characterized as $P_{st}^k=s \xrightarrow{c_1} x_1 \xrightarrow{c_2} x_2 \xrightarrow{c_3} \dots \xrightarrow{c_{k-2}} x_{k-2} \xrightarrow{c_{k-1}} x_{k-1} \xrightarrow{c_p} t $, where intermediate nodes $x_i, i=1, \dots, k-1$ represent nodes at which an alternation of color in path $P_{st}^k$, and the symbol $\xrightarrow{c}$ therefore indicates a maximal sub-path of path $P_{st}^k$ in which all edges are of color $c$. Thus, there exists a path $P'_{st}=s x_1 x_2 \dots x_{k-2} x_{k-1} t$ in $G'$ which corresponds to $P_{st}^k$. It is easy to see that $P'_{st}$ is a shortest path from $s$ to $t$ in $G'$. Indeed, the existence of another, shorter path from $s$ to $t$ in $G'$ and of a corresponding path to it in $G$ would contradict the optimality of $P_{st}^k$ in $G$. Conversely, any shortest path from $s$ to $t$ in $G'$ would induce a corresponding $k$-interchange-constrained shortest path between $s$ and $t$ in $G$, using the same reasoning.
\end{proof}

\section{Some Limitations of Classical Network Indicators}
\label{4}

For a given $k$, $dc_k$ and $di_k$ define two mathematical sequences on $k$. n this section, we illustrate that both the values of these sequences as well as their variations for different values of $k$ are indicative of a network's topological development over time, if we assume this development to be a series of decisions consisting in creating new transit lines or extending existing ones. Moreover, we show through examples that classical network indicator fail to capture this information. \\

We argue that, \textit{ceteris paribus}, the $k$-interchange-constrained diameter mathematical sequence is indicative of the appropriateness of successive additions of lines, that is the quality of past decisions regarding the trade-off \cite{jubilee} between extending existing lines (which would possibly increase the value of the diameter of the graph), versus creating new lines  (which would possibly result in a significant increase of the value of sequence $di_k$ for two consecutive values of $k$, i.e. before and after the line addition). Ideally, new lines would add spread, without extending the diameter.\\

In other words and all things equal, better networks should present, a small variability in their $k$-interchange-constrained diameters for increasing values of $k$. Indeed, small and stable values of $di_k$, for successive values of $k$ indicate a good compromise between the two conflicting feats of not extending existing lines too much, thus extending passenger travel time (which would result in a large value of the $k$-interchange-constrained diameter for some $k$), and only adding new lines when strictly necessary to cover new territory (which would result in an important increase in the value of the $k$-interchange-constrained diameter, for two consecutive values of $k$). Conversely, a significant increase of the $k$-interchange-constrained diameter for two consecutive values of $k$ may indicate the unnecessary addition of a new line to the network, at some point of the decision-making process, when extending an existing line may have been \textit{theoretically} more indicated. The term \textit{theoretically} is highlighted in the previous sentence, because a significant increase in the value of $di_k$, brought by adding a line, may not be the result of poor decision-making and may be simply explained by technical/engineering difficulties to extend lines \cite{jubilee}, or the commercial necessity to create new lines with a special pricing system (e.g. transit connections to airports such as Bangkok's \textit{Airport Link} and Paris' \textit{OrlyVal}). However, from a passenger's point of view, large increases of $di_k$ for successive values of $k$ remain undesirable, as they mean that the coverage of new territory during the development of the network comes with the inconvenience of having to make more transfers.    \\
\begin{figure}
\begin{center}

\includegraphics[width=\linewidth]{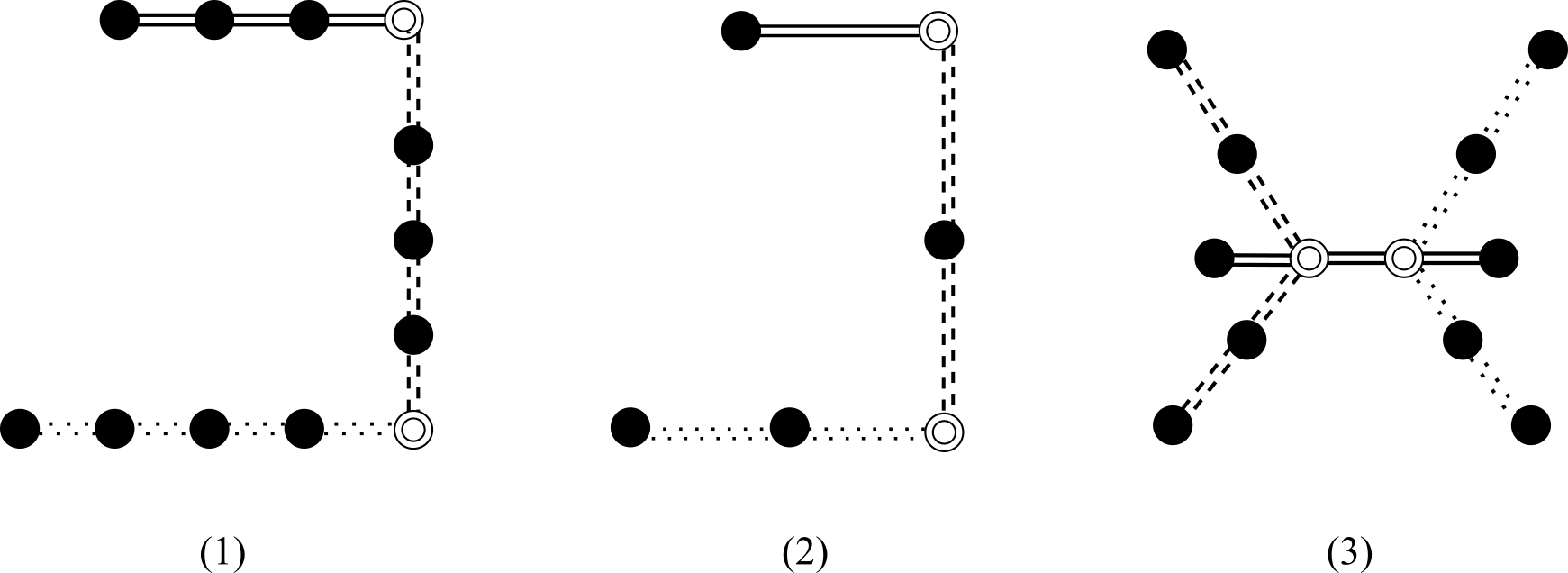}
\caption{\label{toy} First generic example} 
\end{center}
\end{figure}

Figure \ref{toy} further illustrates the intuition behind this indicator on a generic example. It represents three transit networks, denoted $(1)$, $(2)$ and $(3)$, each of which is constituted of three transit lines (distinguished by the representation of their edges as plain, dashed or dotted double-lines) and two transfer stations (represented as empty white nodes), the full black nodes representing non-transfer stations. Network $(3)$ possesses the same number of stations (equal to 12) as network $(1)$ and the same diameter (equal to 5) as network $(2)$. Table \ref{toytable} presents the successive values of $di_k$ as well as its variance, for each of these three networks. Networks $(1)$ and $(2)$ are simple chains. All their stations could have been covered by a single transit line, although it may have been made impossible to do so because of the technical difficulty of extending a line, or by commercial necessities, as previously stated. It remains that, from a purely topological perspective, the use of three different transit lines on these networks was unnecessary. In network $(3)$ however, and assuming a line-by-line development of the network, the definition of three different transit lines was an absolute topological necessity. This network also has a better spread than network $(1)$, for an equal length/number of stations covered and a smaller diameter. Additionally, it covers more stations than network $(2)$, for an equal diameter (i.e. longest possible trip). These facts are reflected by the relatively smaller variance of its $k$-interchange-constrained diameter, as presented in Table \ref{toytable}.

\begin{table}[!ht]

     \begin{center}

   \begin{tabular}{ *{5}{|c}|} 
   \hline
   \backslashbox{Network}{$k$} & $1$ & $2$ & $3$ & Variance\\
   \hline \hline
$(1)$& 4&8&11&\textbf{8.22}\\\hline
$(2)$& 2&4&5&\textbf{1.55}\\\hline
$(3)$& 4&5&5&\textbf{0.22}\\\hline
\end{tabular}
          \end{center}
           \caption{\label{toytable} $k$-interchange-constrained diameters of the illustrative networks in the first generic example}
\end{table}

The second generic example shows that the $k$-interchange-constrained diameter contains an information about network, that is absent in their $\Pi$ and $\tau$ indicators. The two networks presented in Figure \ref{toy2}, which uses the same graphical nomenclature as Figure \ref{toy}, are both constituted of three transit lines and are of equal diameters $d=3$, equal total lengths $|E|=5$ and thus present the same value of network extension $\Pi=\frac{5}{3}=1.66$. Additionally, in both these network, the maximum number of transfers to cover the diameters of the networks equals $2$. Thus, they present the same value of directness $\tau=\frac{3}{2}=1.5$.\\

Although, networks $(4)$ and $(5)$ are topologically identical, network $(5)$ is better from a passenger's perspective, as traveling the whole length of its diameter can be done using a single transit line (but also with one or two interchanges), whereas it cannot be done without at least two interchanges in network $(4)$. This is reflected by the fact that network $(5)$ has a constant $k$-interchange-constrained diameter (i.e. line additions add spread without making the diameter longer), when the $k$-interchange-constrained diameter of network $(5)$ presents a variance of $0.22$.\\

\begin{figure}
\begin{center}
\includegraphics[width=.85\linewidth]{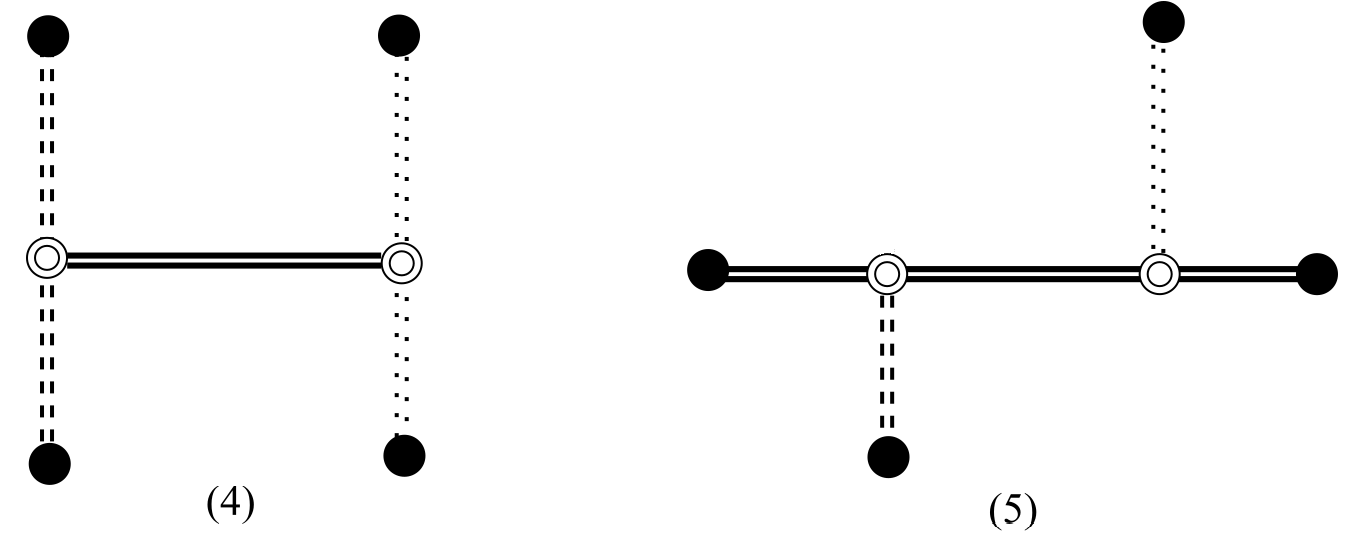}
\caption{\label{toy2} Second generic example} 
\end{center}
\end{figure}

\section{Conclusion}
\label{d}
This paper studied two variants of the shortest path problem in a valued and edge-colored graphs. Given an integer $k$, the \textit{$k$-color-constrained} and the \textit{$k$-interchange-constrained} shortest path problems, respectively seek a shortest path that uses no more than $k$ colors and one that makes no more than $k-1$ alternations of colors. We have shown that the former problem is NP-hard, when the latter is tractable. The study of these problems was motivated by some limitation in the use of diameter-based metrics as structure indicators for transit networks, namely that they do not account for the number of interchanges of lines, which can be a defining factor for the travel convenience of passengers. \\

Thus, we have proposed a new network indicator, the $k$-interchange-constrained diameter of a graph, whose value accounts for travel convenience when measuring the diameter of a graph. Moreover, the stability of this indicator for different values of $k$ is indicative of the effects of successive transit line extensions and new lines creations on improving connectivity, by increasing a network's spread without significantly increasing its diameter. \\

Preliminary analyses we have conducted on three real transit networks, which will be published fully in a future more application-oriented paper, tend to confirm an intuition we have that mature, well connected transit networks (e.g. Paris, Moscow) present rather stable values of their k-interchange-constrained diameter, for variations of $k$, when a more recently established, still rapidly-expanding transit network such as the BTS and MRT network of the Bangkok Metropolitan Area present relatively important fluctuations for this indicator, with successive values of $k$. Thus variations of the $k$-interchange-constrained diameter of a transit network seem to, unsurprisingly, be a reflection of the stage of development it finds itself at.


\bibliographystyle{model1-num-names}

\end{document}